\documentclass[12pt, a4paper]{article}

\usepackage{a4wide}
\usepackage[T2A]{fontenc}
\usepackage[utf8]{inputenc}
\usepackage{amssymb, amsthm}
\usepackage{amsmath}
\usepackage{comment}
\usepackage{color}
\usepackage[pdftex]{graphicx}
\usepackage{hyperref}
\usepackage{mathtools}

\DeclareMathOperator{\sgn}{sgn}

\DeclareMathOperator{\poly}{poly}

\begin{document}

\def\newstr{\par\noindent}
\def\ms{\medskip\newstr}

\renewcommand{\proofname}{Proof}
\renewcommand{\labelenumi}{$\bullet$}

\newcommand{\ifdraft}[1]{#1}
\definecolor{aocolour}{rgb}{0.7,0.8,1}
\definecolor{vmcolour}{rgb}{1,0.8,0.7}
\newcommand{\ao}[1]{\ifdraft{\noindent\colorbox{aocolour}{A.O.: #1}}}
\newcommand{\vm}[1]{\ifdraft{\noindent\colorbox{vmcolour}{V.M.: #1}}}

\newcommand{\Z}{\mathbb{Z}}
\newcommand{\N}{\mathbb{N}}
\newcommand{\R}{\mathbb{R}}
\newcommand{\Q}{\mathbb{Q}}
\newcommand{\K}{\mathbb{K}}
\newcommand{\Cm}{\mathbb{C}}
\newcommand{\Pm}{\mathbb{P}}
\newcommand{\Zero}{\mathbb{O}}
\newcommand{\F}{\mathbb{F}_2}
\newcommand{\ilim}{\int\limits}
\newcommand{\impl}{\Rightarrow}
\newcommand{\set}[2]{\{ \, #1 \mid #2 \, \}}

\newcommand{\A}{\mathcal{A}}
\newcommand{\B}{\mathcal{B}}
\renewcommand{\C}{\mathcal{C}}

\theoremstyle{plain}
\newtheorem{thm}{Theorem}
\newtheorem{lm}{Lemma}
\newtheorem*{st}{Statement}
\newtheorem*{prop}{Property}
\newtheorem{prob}{Problem}
\newtheorem{idea}{Idea}
\newtheorem{conjecture}{Conjecture}

\newtheorem{oldtheorem}{Theorem}
\renewcommand{\theoldtheorem}{\Alph{oldtheorem}}
\newtheorem{oldconj}{Conjecture}
\renewcommand{\theoldconj}{\Alph{oldconj}}

\theoremstyle{definition}
\newtheorem{defn}{Definition}
\newtheorem{ex}{Example}
\newtheorem{cor}{Corollary}

\theoremstyle{remark}
\newtheorem*{rem}{Remark}
\newtheorem*{note}{Note}

\title{Cocke--Younger--Kasami--Schwartz--Zippel algorithm and relatives}

\author{Vladislav Makarov\thanks{Saint-Petersburg State University. Supported by Russian Science Foundation, project 18-11-00100.}}

\maketitle

\begin{abstract}
The equivalence problem for unambiguous grammars is an important, but very difficult open question in formal language theory.
Consider the \emph{limited}
equivalence problem for unambiguous grammars --- for two unambiguous grammars $G_1$ and $G_2$, 
tell whether or not they describe the same set of words of length $n$. Obviously, the naive approach requires exponential time with respect to $n$.
By combining two classic algorithmic ideas, I introduce a $O(\poly(n, |G_1|, |G_2|))$ algorithm for this problem. Moreover, the ideas behind
the algorithm prove useful in various other scenarious.
\end{abstract}

\section{Preface}

This Section contains some details about the technical structure of this paper
and the reasons for its existence. Therefore, feel free to skip it. Just remember that this
paper has a ``prequel'': ``Why the equivalence problem for unambiguous grammars 
has not been solved back in 1966?''.

This paper and the companion paper ``Why the equivalence problem for unambiguous grammars 
has not been solved back in 1966?' are based on the Chapters 4 and 3 of my Master's thesis~\cite{makarov-master} respectively.
While the thesis is published openly in the \href{https://dspace.spbu.ru/handle/11701/30103?locale=en}{SPbSU system}, it has not been published
in a peer-reviewed journal (or via any other scholarly accepted publication method) yet.

These two papers are designed to amend the issue. As of the current date, I have not
submitted them to a refereed venue yet. Therefore, they are only published as arXiv
preprints for now. 

Considering the above, it should not be surprising that huge parts of the original text are copied almost
verbatim. However, the text is not totally the same as the Chapter 4 of my thesis. Some things
are altered for better clarity of exposition and there are even some completely new parts.

Why did I decide to split the results into two papers? There are two main reasons. 

Firstly, both papers are complete works by themselves. From the idea standpoint, some
of the methods and results of this paper were motivated by the careful
observation of results of the companion paper. However, the main result of this paper
is stated and proven without any explicit references to the content of the companion paper.

The second reason is explained in the ``prequel'' paper in great detail. For simplicity, let us
say that this paper contains specific ``positive'' results, while the previous focuses
on more vague ``negative'' results.

With technical details out of the way, let us move on to the mathematical part.

\section{Cocke--Younger--Kasami--Schwartz--Zippel algorithm}\label{section_cyksz}

Recall the statement of the equivalence problem for unambiguous grammars.

\begin{prob}[The equivalence problem for unambiguous grammars.] \label{grail}
You are given two ordinary grammars $G_1$ and $G_2$. Moreover, you know  that
they are both unambiguous from a $100\%$ trustworthy source. Is there 
an algorithm to tell whether $L(G_1)$ and $L(G_2)$ are equal? Because the grammars
are guaranteed to be unambiguous, the algorithm may behave arbitrarily if either of $G_1$
and $G_2$ is ambiguous, including not terminating at all.
\end{prob}

In the ``prequel'' paper we studied the extents and limitations of matrix substitution.
Now, let us focus on another approach. 

\begin{defn} For a language $L$, its \emph{$n$-slice} is the language $\set{w}{w \in L, |w| = n}$ of all words from $L$ of length exactly $n$.
\end{defn}

As we have seen in the ``prequel'', matrix substitution is in some way ``independent'' over the slices of
$L(G_1)$ and $L(G_2)$: $L(G_1)$ and $L(G_2)$ are $d$-similar if and only if $n$-slices
of $L(G_1)$ and $L(G_2)$ are $d$-similar for all $n$. 

The main advantage of the matrix substitution approach is its \emph{uniformity}: it checks
all $n$-slices for similarity at the same time. However, we can do better if the uniformity is not
required. Specifically, consider the following problem:

\begin{prob}\label{slices} Given two unambiguous grammars $G_1$ and $G_2$, tell whether $n$-slices of $L(G_1)$ and $L(G_2)$ are equal. 
\end{prob}

The naive solution for Problem~\ref{slices} takes $\Theta(|\Sigma^n| \cdot \poly(n, |G_1|, |G_2|))$ time. Fortunately, it is unnecesarry to iterate over all strings:

\begin{thm}[Cocke--Younger--Kasami--Schwartz--Zippel algorithm]\label{five_guys_algo} If both $G_1$ and $G_2$ are in 
Chomsky normal form, it is possible to solve Problem~\ref{slices} in $O(n^3 \cdot (|G_1| + |G_2|))$ time with bounded one-sided error from randomization.
\end{thm}
\begin{rem} Of course, Theorem~\ref{five_guys_algo} cannot help with Problem~\ref{grail} at all, as long as we are interested in \emph{decidability only}. However, the running time improvement over the naive
algorithm is drastic, making iterating over all small $n$ a viable heuristic for Problem~\ref{grail}. 
\end{rem}

The rest of this Section is dedicated to the proof of Theorem~\ref{five_guys_algo}, and, even more
importantly, overview of the main ideas. The next several paragraphs only provide motivation for
the ideas that are used in the proof, so I will keep details a bit vague in them. Moreover, these paragraphs
use the ideas from the ``prequel'' paper; feel free to skip them, as the proof itself should also be clear enough.

Let us look back at the Amitsur--Levitsky identity. As I noted before, it looks like the definition
of determinant. To get \emph{exactly} the definition of determinant, we need to add another index
to each $X_i$:
\begin{align}
	\label{before}
	 \text{before: } &\sum\limits_{\sigma \in S_{2d}} (-1)^{\sgn(\sigma)}
	X_{\sigma(1)} X_{\sigma(2)} \ldots X_{\sigma(2d)}   \\
	\label{after}
	\text{after: } &\sum\limits_{\sigma \in S_{2d}} (-1)^{\sgn(\sigma)}
	X_{1, \sigma(1)} X_{2, \sigma(2)} \ldots X_{2d, \sigma(2d)} 
\end{align}

While the Expression~\eqref{before} is a polynomial identity for $d \times d$ matrices, 
the Expression~\eqref{after} is not, even for $1 \times 1$ matrices (that is, scalars). Indeed, there \emph{are} some matrices
with non-zero determinant.

Now, suppose that we have two unambiguous grammars $G_+$ and $G_-$ over an alphabet 
$\Sigma \coloneqq \{a_1, a_2, \ldots, a_n\}$, generating even and odd permutations of length $2d$ respectively.
More formally, $L(G_+) = \set{a_{\sigma(1)} a_{\sigma(2)} \ldots a_{\sigma(2d)}}{\sigma \in S_{2d}, \sgn(\sigma) = +1}$ and $L(G_+) = \set{a_{\sigma(1)} a_{\sigma(2)} \ldots a_{\sigma(2d)}}{\sigma \in S_{2d}, \sgn(\sigma) = -1}$. Then, by Amitsur--Levitsky theorem,
we cannot tell them apart by substituting matrices of size $d \times d$ and smaller. Of course, we can pick
larger matrices or bash this specific instance of the Problem~\ref{grail} in many other ways, because both languages are finite. But it is not the point, so let us not do that.

In the ``prequel'', we had no choice except for representing everything
implicitly via equations and rely on Tarski--Seidenberg theorem to rule
everything out. There are two reasons for that: we wanted everything to work for all matrices
of a small norm and we wanted to handle all $n$-slices at the same time. If we know the matrices
beforehand and want to consider only a specific slice (the $(2d)$-th slice in our case), we can compute the values of $f$ with dynamic programming. And if we somehow manage incorporate the second indices (as in Expression~\eqref{after}) in the dynamic programming, then there will not be any systematic problem
with polynomial identities (accidental coincidencies can still happen, though).

As it turns out, to solve Problem~\ref{slices}, we need only these two ideas (an \emph{indexation trick},
as I call it, and a somewhat careful dynamic programming), but we do not need to substitute matrices
anymore, just numbers will be enough. These ideas will be crucial through the rest of the text, so
I recommend to understand them well.

\begin{proof}[Proof of Theorem~\ref{five_guys_algo}, assuming Theorem~\ref{eval}.]
Specifically, consider the following mapping $f$ from subsets of $\Sigma^n$ to polynomials with integer
coefficients in $|\Sigma| \cdot n$ variables $x_{a, i}$, where $a$ ranges over $\Sigma$ and $i$
ranges over $[1, n]$:  a single word $w_1 w_2 \ldots w_n$ maps to a monomial $f(w) \coloneqq x_{w_1, 1} \cdot x_{w_2, 2} \cdot \ldots \cdot x_{w_n, n}$ and a language $L \subset \Sigma^n$ maps to 
$f(L) \coloneqq \sum\limits_{w \in L} f(w)$. Here, I abuse notation a little bit, because $f$ is actually two different mappings: one is from $\Sigma^n$ to monomials and the other is from subsets
of $\Sigma^n$ to polynomials. I hope that this does not cause confusion. 

It is important that these polynomials are normal, commutative polynomials, no trickery here (at least in the proof of Theorem~\ref{five_guys_algo}). And there is no need to substitute matrices in place
of variables, because the indexing itself takes care of noncommutativity of word concatenation. For example, 
$f(ab) = x_{a, 1} x_{b, 2}$, but $f(ba) = x_{b, 1} x_{a, 2} = x_{a, 2} x_{b, 1} \neq x_{a, 1} x_{b, 2}$. Hence, different subsets
of $\Sigma^n$ have different images under $f$: if a word $w$ is present in $K$, but not in $L$,
then the monomial $f(w) = x_{w_1, 1} x_{w_2, 2} \ldots x_{w_n, n}$ is present in $f(K)$, but not in $f(L)$. Moreover, all coefficients of $f(K)$ are either $0$ or $1$, depending on whether the corresponding word belongs to $K$ or not.

In particular, if $L_1$ and $L_2$ are $n$-slices of $L(G_1)$ and $L(G_2)$ respectively, then $f(L_1) = f(L_2)$ if and only if $L_1 = L_2$. Moreover, both $f(L_1)$ and $f(L_2)$ are polynomials of degree at most $n$, because they both are sums of several monomials of degree exactly $n$.
Therefore, we can use Schwartz-Zippel lemma here to compare them without computing them explicitly.

To be exact, let us fix some large finite field $F$. Because all coefficients of $f(L_1)$ and $f(L_2)$
are zeroes or ones, $f(L_1)$ and $f(L_2)$ are equal as integer polynomials if and only
if they are equal as polynomials over $F$. Evaluate them in 
a random point from $F^{|\Sigma| \cdot n}$ by the following Theorem~\ref{eval}. By Schwartz-Zippel lemma, the probability of a false positive (polynomials are not equal, but their values are) is  at most $\frac{n}{|F|}$, which is at most $1/2$ when $|F| \geqslant 2n$. Of course, there are no false negatives. Repeat several times to obtain the desired error probability. 
\end{proof}

\begin{thm}\label{eval} Given an unambiguous grammar $G$ in Chomsky normal form and a field $F$, it is possible to evaluate $f(K)$ in a point $x \in F^{|\Sigma| \cdot n}$ in $O(n^3 \cdot |G|)$ field operations, 
where $K$ is the $n$-th slice of $L(G)$.
\end{thm}
\begin{proof} The proof is an adaptation of standard Cocke--Younger--Kasami cubic parsing algorithm.
For a non-empty subsegment $[\ell, r)$ of $[1, n + 1)$ and a language $K \subset \Sigma^{r - \ell}$, 
define $f_{[\ell, r)} \coloneqq \sum\limits_{w \in K} x_{w_1, \ell} x_{w_2, \ell + 1} \ldots x_{w_{r - \ell}, r - 1}$. This is a generalisation of $f$. Indeed, $f_{[1, n + 1)}(K) = f(K)$ when $K \subset \Sigma^n$, and, therefore, both sides of the equation are defined.

Then, $f_{[\ell, r)} (K \sqcup L) = f_{[\ell, r)} (K) + f_{[\ell, r)} (L)$ and $f_{[\ell, r)}(KL) = f_{[\ell, m)} (K) f_{[m, r)} (L)$, as long as everything is defined. Note that there is no requirement for the concatenation $KL$ to be unambiguous, because it is umambiguous automatically. Indeed, all words in $K$ have length $m - \ell$ and all words in $L$ have length $r - m$. Hence, for any word $w$ in $KL$, there is only one way to represent is a concatenation of a word from $K$ and $L$: $w_1 \ldots w_{m - \ell} \in K$ and $w_{m - \ell + 1} \ldots w_{r - \ell} \in L$.

Hence, for a nonterminal $C$ of $G$ and a non-empty subsegment $[\ell, r)$ of $[1, n + 1)$,
\begin{equation}
f_{[\ell, r)}(L(C)) \coloneqq \begin{cases} \sum\limits_{(C \to a) \in R} x_{a, \ell}, &\text{if $r - \ell = 1$,} \\
\sum\limits_{(C \to DE) \in R} \sum\limits_{m=\ell+1}^{r-1} f_{[\ell, m)}(L(D)) f_{[m, r)} (L(E)), &\text{otherwise.} \end{cases}
\end{equation}
In the second of the above cases, we implicitly use the unambiguity of concatenation $L(D) L(E)$ by iterating over $m$. 

Now, let us evaluate $f_{[\ell, r)}(L(C))$ in the point $x$. 
To do so, apply the above formulas in the order of increasing $r - \ell$.
In the end, the value of $f_{[1, n+1)}(L(S))$ in $x$ is 
exactly what we needed to compute. The total number of feild operations is $O(n^3 \cdot |R|) = O(n^3 \cdot |G|)$: for each
triple $\ell < m < r$, we iterate over all ``normal'' rules of the grammar.
\end{proof}
\begin{rem} It is possible to check whether $L(G_1)$ and $L(G_2)$ have a difference of length \emph{at most} $n$ (rather then \emph{exactly} $n$, as in Theorem~\ref{five_guys_algo}) with the same running time. To do so, run the above procedure, but compare
the sums $\sum\limits_{i=2}^{n+1} f_{[1, i)}(L(G_1))$ and $\sum\limits_{i=2}^{n+1} f_{[1, i)}(L(G_2))$
instead of $f_{[1, n+1)}(L(G_1))$ and  $f_{[1, n+1)}(L(G_2))$.
\end{rem}

Let us look at the polynomial $f$ more closely. As a polynomial with integer coefficients, it can be seen as a function from $\Z^{|\Sigma| \cdot n}$ to $\Z$. An \emph{input} for $f$ is any assignment of integer values to variables $x_{a, i}$, where
$a$ ranges over $\Sigma$ and $i$ ranges over $[1, n]$. Intuitively, the variable $x_{a, i}$ represents
``to which extent'' the letter $a$ appears in the $i$-th position. There are two special kinds of inputs:
null-like and word-like.
\begin{defn} An input $x$ is \emph{null-like} if there exists an index $i$ from $1$ to $n$, such that
$x_{a, i} = 0$ for all $a \in \Sigma$.
\end{defn}
\begin{defn} An input $x$ is \emph{word-like} if for all indices $i$ from $1$ to $n$ there is \emph{exactly one} such $a \in \Sigma$, that $x_{a, i} = 1$ and $x_{b, i} = 0$ for all $b \in \Sigma \setminus \{a\}$.
\end{defn}

Null-like inputs are not interesting: because all $x_{a, i} = 0$ for all $a \in \Sigma$, the $i$-th variable
in each monomial of $f$ evaluates to $0$. Hence, each monomial evaluates to $0$ and so does $f$.

Word-like inputs are more important. A word-like input is a renundant representation of a word:
instead of just $n$ letters, we have $|\Sigma| \cdot n$ answers to questions ``is there a letter $a$ on the $i$-th position?''. And, fittingly, there is exactly one letter on each position. 

Consider a word-like input $x$ corresponding to a word $w$. Clearly, on the input $x$, the monomial $x_{u_1, 1} x_{u_2, 2} \ldots  x_{u_n, n}$ evaluates to $1$ if $w = u$ and to $0$ otherwise. Hence, the value of $f(K)$ on a
word-like input simply shows whether the corresponding word lies in $K$ or not. This is just 
the parsing problem for the grammar $G$ and the word $w$.

In fact, under a closer examination (I will not dwelve in the details), 
the proof of Theorem~\ref{eval} becomes \emph{exactly}
the standard cubic-time parsing algorithm when we consider only word-like
inputs $x$. 

In a sense, the whole approach has even remote chances of working \emph{exactly} because
we consider a much larger set of inputs, not only word-like ones. Restricting ourselves only 
to word-like input would be the same as picking random \emph{words} of length $n$ and checking
whether they belong to $L_1$ and $L_2$. Picking random \emph{words} works poorly, for example, if $L_1$ and $L_2$ are obtained by explicitly
adding two different words to the same unambiguous grammar (and in the myriad of more complicated situations).

The above observation is extremely important in the following Sections. We will see its ``positive''
side in Section~\ref{section_monotone} and its ``negative'' side in Section~\ref{section_conjunctive}.

\section{Relations to circuit complexity}\label{section_monotone}

Of course, there is a more general idea in the proof of Theorem~\ref{five_guys_algo}:

\begin{idea}\label{idea} To prove something about $n$-slices of $L(G)$, where $G$ is some type of grammar (unambiguous, ordinary, GF(2), et cetera), define $f(K)$ in terms of the operations that are ``inherent'' to the corresponding grammar formalism. Evaluate $f(L(G))$ similarly to Theorem~\ref{eval} and somehow exploit the fact that evaluating $f(L(G))$ is a much more general problem than string parsing for $G$.
\end{idea}

For example, let me sketch the proof of the following result:
\begin{thm}\label{five_guys_gf2} Given a GF(2)-grammar $G$, it is possible to tell whether $n$-slice of $L(G)$
is empty in $O(n^3 \cdot |G|)$ time with randomization and bounded one-sided error.
\end{thm}
\begin{proof}[A sketch.] The proof is mostly the same. Let us focus on the differences. The first difference is that we define $f$ and $f_{[\ell, r)}$ as polynomials over $\mathbb{F}_2$, not $\Z$, because we have symmetric difference and GF(2)-concatenation instead of disjoint union and unambiguous concatenation now. 

The second difference is that $F$ has to be a field of characteristic $2$ (otherwise the notion 
of evaluating $f(K)$ on a point of $F$ makes no sense). Again, this is not a problem: for each $b$,
there is a field of characteristic $2$ and size $2^b$.
\end{proof}
\begin{rem} Of course, Theorem~\ref{five_guys_algo} immediately follows from Theorem~\ref{five_guys_gf2}, 
but I think that the ideas are easier to understand when they are illustrated on the proof of Theorem~\ref{five_guys_algo}.
\end{rem}

What is, in my opinion, much more interesting and important, is that we can use Idea~\ref{idea}
to prove some \emph{lower bounds} style results. The general ``line of attack'' is like this:
\begin{enumerate}
\item[1.] Consider a language $L_{\textrm{orig}}$ over an alphabet $\Sigma$. Assume that there is a grammar $G$ of some type (unambiguous, ordinary, GF(2), et cetera) for $L_{\textrm{orig}}$. 
\item[2.] By applying standard grammar closure properties (closure under certain types of transductions, set operations, et cetera), obtain a faimly $L_n$ of languages, each over its own
alphabet $\Sigma_n$, such that a grammar $G$ for $L$ leads to a grammar $G_n$ of size 
$\poly(n, |G|)$ for $L_n$. The alphabets $\Sigma_n$ can depend on $n$ and even grow in size
polynomially with $n$. 
\item[3.] Use the Idea~\ref{idea} to evaluate some kind of function on $|\Sigma| \cdot n$ variables
with some kind of circuit of size $\poly(n, |G_n|) = \poly(n, |G|)$. The exact nature of the variables, the function and the circuit depends on the type of grammar.
\item[4.] Use some kind of circuit lower bound (either already known or a conjectured one) to show that this function cannot be evaluated
by a small circuit of a small size.
\item[5.] Hence, there was no grammar for $L_{\textrm{orig}}$ in the first place.
\end{enumerate}

The exact details of the plan are flexible. For example, sometimes it is possible to omit the first step altogether and still get some kind of meaningful result. Moreover, if the circuits on the step 4 are unrestricted boolean circuits and step 3 takes polynomial time, then we can replace ``circuit lower bound'' by a ``algorithm running time lower bound'' in the step 4, because all circuits here are generated by a uniform polynomial-time procedure. In principle, it is possible to replace all mentions
of ``polynomial time'' with explicit functions that may not even be polynomial (say, like, $2^{n/10}$), but I do not yet know any situation where that would significantly help.

To see, how the plan would work for ordinary grammars, consider the language $P_n$ of
all permutations of length $n$ over the alphabet $\Sigma_n = \{1, 2, \ldots, n\}$ (yes, the digits are
the letters here). For example, $P_3 = \{123, 132, 213, 231, 312, 321\}$. It is already known that
the size of the smallest ordinary grammar for $P_n$ grows with $n$ exponentially~\cite[Theorem 30]{sn_ordinary}. Let us
prove that the growth is superpolynomial with our method.

Let us define $f(K)$ as a boolean function: $f(w) \coloneqq x_{w_1, 1} \wedge x_{w_2, 2} \wedge \ldots \wedge x_{w_n, n}$ and $f(K) \coloneqq \bigvee\limits_{w \in K} f(w)$. Why a boolean function? Because set union and concatenation are expressed in terms of disjunction and conjunction.
Such a definition works well with concatenation and set union, which correspond to the conjunction and disjunction respectively.

Specifically, $f_{[\ell, r)}(K \cup L) = f_{[\ell, r)} (K) \vee f_{[\ell, r)} (L) $ and $f_{[\ell, r)} (KL) = f_{[\ell, m)}(K) \wedge f_{[m, r)}(L)$, as long as everything is defined.
The latter is obvious, so let me illustrate the former on an example:
$f_{[1, 3)}(\{ab, aa\}) \wedge f_{[3, 5)}(\{cc, ad\})) = ((x_{a, 1} \wedge x_{b, 2}) \vee (x_{a, 1} \wedge x_{a, 2})) \wedge ((x_{c, 3} \wedge x_{c, 4}) \vee (x_{c, 3} \wedge x_{c, 4})) = 
((x_{a, 1} \wedge x_{b, 2} \wedge x_{c, 3} \wedge x_{c, 4}) \vee (x_{a, 1} \wedge x_{a, 2} \wedge x_{a, 3} \wedge x_{c, 4}) \vee
(x_{a, 1} \wedge x_{a, 2} \wedge x_{c, 3} \wedge x_{c, 4}) \vee (x_{a, 1} \wedge x_{a, 2} \wedge x_{a, 3} \wedge x_{d, 4})) = f_{[1,5)}(\{abcc, abad, aacc, aaad\}) = f_{[1, 5)} (\{ab, aa\} \cdot \{cc, ad\})$ by distributivity properties.

In the end, we get a \emph{monotone} (only disjunctions and conjunctions) boolean circuit
of size $O(n^2 \cdot |G_n|)$ that computes $f(P_n) = \bigvee\limits_{\sigma \in S_n} (x_{\sigma(1), 1} \wedge x_{\sigma(2), 2} \ldots \wedge x_{\sigma(n), n})$. In other words, a monotone circuits
that checks whether there is a perfect matching in a bipartite graph with $n$ vertices in each part. 
(the variables $x_{i, j}$ correspond to the presence or absence of an edge between the $i$-th vertex on the left and the $j$-th vertex on the right). Razborov~\cite{razborov} proved that this problem requires monotone
circuits of size $n^{\Omega(\log n)}$. Hence, $|G_n|$ grows superpolynomially. 

To get a grammar nonexistence result for a specific language (as I said, this is often unnecessary), encode the letter $i$ in $\Sigma_k$ by the word $w_{k, i} \coloneqq 0^{i-1} 1 0^{k-i}$ over an alphabet $\Sigma = \{0, 1\}$. Now, define $P_{\textrm{orig}}$ as the  union of $P_{\textrm{base}} \coloneqq \bigcup\limits_{k=0}^{+\infty} \set{w_{k, \sigma(1)} w_{k, \sigma(2)} \ldots w_{k, \sigma(k)}}{\sigma \in S_k}$ and \emph{any} set $P_{\textrm{extra}}$ of words that do not decode to anything. That is, $P_{\textrm{extra}}$ can be any subset of $\{0,1\}^* $ that does not intersect 
$\bigcup\limits_{k=0}^{+\infty} \set{w_{k,i}}{1 \leqslant i \leqslant n}^*$. Then, there is no ordinary grammar for $P_{\textrm{orig}}$, independently of the choice of the $P_{\textrm{extra}}$.

Indeed, to get $P_n$ from $P_{\textrm{orig}}$ we need to intersect with $\Sigma^{n^2}$ (to leave only the words of correct length) and decode the letters by replacing $w_{n, i}$ with $i$ (this can be done by a deterministic transducer). Because $P_{\textrm{extra}}$ specifically can contain only
words that do not successfully decode to anything, the resulting language will be exactly $P_n$.
Moreover, each of the two steps blows up the grammar size only by a polynomial in $n$ factor.

Hence, there is no ordinary grammar for $P$ (grammar for $P$ implies small grammars for $P_n$, and that in turn implies small monotone circuits for perfect matching). To be honest, this is not a particularly interesting (or new) result, but it does a good job of highlighting the steps of the plan.

\section{Can we do the same with conjunctive grammars?}\label{section_conjunctive}

In the previous Section we have seen how the indexing trick helps with lower-bound style
results for unambiguous and even arbitrary ordinary grammars. Can we do something
similar with conjunctive grammars? This is an interesting question, because there is 
currently no known satisfying methods of proving that some language cannot be described
by a conjunctive grammar, except for the following theorem which is based on the complexity
of standard algorithms for parsing:
\begin{oldtheorem}\label{conj_bound} If $L$ is described by a conjunctive grammar, then
$L \in \mathrm{DTIME}(O(n^3)) \cap \mathrm{DSPACE}(O(n))$.
\end{oldtheorem}

Unfortunately, no. 
Indeed, while there is no small monotone circuits for bipartite matching, \emph{there is} a polynomially-sized conjunctive grammar for $P_n$. Moreover, $P_n$ is an intersection of several regular languages: $P_n = \Sigma^n \cap \bigcap\limits_{i=1}^n (\Sigma^* i \Sigma^*)$, where $\Sigma = \{1. 2, \ldots, n\}$. 
So, what does go wrong? In short, everything.

Because conjunctive grammars are an extension of ordinary grammars, we need to define $f$ the same way as we 
did for ordinary grammars: $f(w_1 \ldots w_n) \coloneqq (x_{w_1, 1} \wedge x_{w_2, 2} \wedge \ldots \wedge x_{w_n, n})$ and
$f(L) \coloneqq \bigvee\limits_{w \in L} f(w)$. Then, the concatenation and the union of languages
will work as expected. Unfortunately, the same is not true for the intersection: 
$f(\{ab\} \cap \{ac\}) = f(\varnothing) = 0 \neq (x_{a, 1} \wedge x_{b, 2} \wedge x_{c. 2}) = (x_{a, 1} \wedge x_{b, 2}) \wedge (x_{a, 1} \wedge x_{c, 2}) = f(\{ab\}) \wedge f(\{bc\})$. 

In fact, not only the conjunction does not
work, but there is no such boolean function $\varphi$, that $f(K \cap L) = \varphi(f(K), f(L))$. 
Indeed, assume the contrary.
Consider an assignment of variables that maps all $|\Sigma| \cdot n$ variables to $1$.
Under such an assignment, $0 = f(\varnothing) = f(\{ab\} \cap \{ac\}) = \varphi(f(\{ab\}), f(\{ac\})) = \varphi(x_{a, 1} \wedge x_{b, 2}, x_{a, 1} \wedge x_{c, 2}) = \varphi(1, 1) = \varphi(x_{a, 1} \wedge x_{b, 2}, x_{a, 1} \wedge x_{b, 2}) = \varphi(f(\{ab\}), f(\{ab\})) = f(\{ab\} \cap \{ab\}) = f(\{ab\}) = x_{a, 1} \wedge x_{b, 2} = 1$, contradiction (all these equation signs correspond to equality in the aforementioned point and not to the equality of functions). 

Hence, there is no way to ``translate'' the intersection of languages to a boolean operation on the corresponding functions. The issue here is the existence of \emph{extra} conjuncts 
like $x_{a, 1} \wedge x_{b, 2} \wedge x_{c, 2}$, which do not correspond to any word, because of a repeated
index (in this case, index $2$). So, if we blindly transform the grammar into a monotone circuit, the resulting circuit will not compute $f(K)$ itself, but, rather, $f(K)$ with some extra conjuncts. 

So, how to deal with extra conjucts? I do not know any promising way, except for considering
only such inputs, that two variables with the same index cannot be mapped to $1$ at the same time
(then, each extra conjunct is automatically mapped to $0$). Or, in  existing terms, by 
considering
word-like inputs only. Null-like inputs are irrelevant, because $f(K)$ always evaluates to $0$ on them.

Then, the rest of argument will proceed without a hitch, leading to the following result: if $G$ is a conjunctive grammar and $K$ is the $n$-slice of $L(G)$, then there is a polynomial-size monotone  boolean circuit that outputs the same values as the monotone boolean function $f(K) = \bigvee\limits_{w \in K} (x_{w_1, 1} \wedge \ldots x_{w_n, n})$ on all word-like inputs.

Unfortunately, this is an extremely weak statement. 
By the above argument, because \emph{there exists} a polynomially-sized conjunctive grammar for $P_n$, there must be a polynomial-sized monotone circuit that works on all word-like inputs.
And, indeed, there is. A boolean \emph{formula}, even: $\bigwedge_{i=1}^n \bigvee_{j=1}^n x_{i, j}$. Informally, this formula checks whether each letter of the alphabet is present somewhere in the word. The uncanny resemblence between this formula and the original conjunctive grammar is not a coincidence.

Indeed, as mentioned in the previous Section, the result of evaluating $f(K)$ on a word-like input corresponding to word $w$, is $1$ when $w \in K$ and $0$ when $w \notin K$. So, if we want to prove that there is no conjunctive grammar for $L_{\textrm{orig}}$, we need to prove that the inclusion 
problem for $L_n$ is harder time- or memory-wise than the cubic parsing algorithm that we used in the reduction!
In the end, our best hope here is arriving to a  \emph{weaker} version of Theorem~\ref{conj_bound} in a \emph{very} roundabout way.
This is dissapointing, but not surprising. As I menitioned before, the Idea~\ref{idea}
works \emph{exactly} because it expands the ``working space'' from only word-like inputs
to a much richer set. So, if we are forced to shrink the space back, we are left with a simple
restatement of Cocke--Younger--Kasami algorithm.

\section{Conclusion}

The above Sections show that the ideas behind CYKSZ algorithm are somewhat universal and can be applied in different and, sometimes, unexpected
ways. Except for the above, I suggest another interesting direction of research. As Theorem~\ref{five_guys_algo} shows, the limited equivalence
is in P for reasonable enough statement of the problem. However, the limited equivalence for ordinary grammars is in NP even when one of the 
grammars describes the whole language. The reason is the same as in proof of undecidability of ordinary grammar equivalence: ordinary
grammars can encode computation histories of Turing machines. Hence, the existence of Theorem~\ref{five_guys_algo} suggest in some vague
way that unambiguous grammars cannot encode computation histories of Turing machines. Therefore, it seems that we have some heuristic
evidence that suggests that equivalence for unambiguous grammars is ``much weaker'' compared to the equivalence for arbitrary ordinary grammars.

\end{document}

